\documentclass[10pt,twocolumn,twoside]{ieeeconf}

\usepackage{amssymb,amsmath}
\usepackage{amsfonts}
\usepackage[english]{babel}
\usepackage[latin1]{inputenc}
\usepackage{fancyhdr}
\usepackage{yfonts}
\usepackage{mathrsfs}
\usepackage[dvips]{graphicx}
\usepackage{enumerate}
\usepackage{xspace}
\usepackage{cite}

\usepackage{psfrag}
\usepackage{url}

\usepackage[dvipsnames]{xcolor}
\usepackage{xcolor}
\usepackage{lipsum}

\usepackage{verbatim}

\providecommand{\com[2]}{\begin{tt}[#1: #2]\end{tt}}
\usepackage{color}

\newcommand{\modif}[1]{{#1}}

\newtheorem{theorem}{Theorem}
\newtheorem{lemma}[theorem]{Lemma}

\newtheorem{proposition}[theorem]{Proposition}

\providecommand{\prt}[1]{\left( #1 \right)}
 
\providecommand{\PP}{\mathbb{P}}

\providecommand{\abs[1]}{\left|#1\right|}

\providecommand{\0}{\textbf{0}}\providecommand{\quotes[1]}{``#1"}
\def\0{{\bf 0}}

\def\N{\mathbb{N}}

\def\E{\mathbb{E}}
\def\DD{\mathcal{D}}
\def\NN{\mathcal{N}}

\newcommand{\ie}{{\it i.e. }}

\newcommand{\ssum}{\displaystyle\sum}
\newcommand{\Var}[1]{\text{Var}{#1}}

\newcommand{\rpm}{\raisebox{.2ex}{$\scriptstyle\pm$}}
\newcommand{\ooneovern}{o\left(\frac{1}{n}\right)}

\begin{document}

\title{Open Multi-Agent Systems:\\ Gossiping with Random Arrivals and Departures}

\author{Julien M. Hendrickx and  Samuel Martin
\thanks{Julien Hendrickx is with the ICTEAM institute, Universit\'e catholique de Louvain, Louvain-la-Neuve, Belgium.  {\tt\small julien.hendrickx@uclouvain.be,} 
Samuel  Martin  is  with  Universit\'e
de Lorraine  and  CNRS,  CRAN,  UMR
7039, 2  Avenue  de  la  For\^et
et de Haye, 54518  Vandoeuvre-l\`es-Nancy,  France
{\tt\small samuel.martin@univ-lorraine.fr}}
\thanks{This work was supported by the Belgian Network DYSCO (Dynamical Systems, Control, and Optimization), 
funded by the Interuniversity Attraction Poles Program, initiated by the Belgian Science Policy Office, and by the ANR project COMPACS ANR-13-BS03-0004 and by projects PEPS MoMIS MADRES and PEPS INS2I CONAS funded by the CNRS.}
}

\maketitle

\begin{abstract}
We consider open multi-agent systems. Unlike the systems usually studied in the literature, here agents may join or leave while the process studied takes place. The system composition and size evolve thus with time. We focus here on systems where the interactions between agents lead to pairwise gossip averages, and where agents either arrive or are replaced at random times. These events prevent any convergence of the system.
Instead, we describe the expected system behavior by showing that the evolution of scaled moments of the state can be
characterized by a 2-dimensional (possibly time-varying) linear dynamical system.
We apply this technique to two cases : (i) systems with fixed size where leaving agents are immediately replaced, and (ii) systems where new agents keep arriving without ever leaving, and whose size grows thus unbounded.
\end{abstract}

\section{Introduction}
Two of the most important features of multi-agent systems are their flexibility and scalability. Accordingly, these systems are expected to cope with agent failures and new agent arrivals. Real life examples of the multi-agent systems with such properties include flock of birds, ad-hoc networks of mobile devices, or the Internet. Social systems of various scales have these properties : work teams in companies or laboratories subject to important turnovers, companies themselves, or even entire countries and their cultural norms.

However, the framework typically used to study formal models of multi-agent systems supposes that, while the system may be complex, its composition remains unchanged over time (albeit the interaction topology can evolve). Under this assumption, researchers are able to characterize the long term behavior of the multi-agent system such as convergence and synchronization.

This apparent contradiction is justified when agents arrivals and departures are sufficiently rare as compared to the time-scale of the process taking place in the system.
In such cases, it makes indeed sense to assume that the composition of the systems remains unchanged while the process takes place.

Nevertheless, the probability of a node failure is expected to grow with the number of agents. As a consequence, for large systems, this constant size assumption no longer holds. Similarly, in some systems such as living systems with birth processes for instance, the probability of a node arrival increases with the system size, so that the constant composition assumption also stops being relevant when the system size is large. Companies or human societies are instances of such systems where the system's growth is proportional to its size.
This assumption may also appear unsuitable in extreme environments, where communication is difficult and infrequent, leading to slow convergence rate, relative to which the agent failure rate may be important.

Hence we consider here \emph{open multi-agent systems}, where agents keep arriving and/or leaving during the execution of the process considered, an example of which is illustrated in Figure~\ref{fig:fixed_system_size_with_replacement_small}.\\
Repeated arrivals and departures result in important differences in the analysis or the design of open multi-agent systems and cause several challenges:

\emph{State dimension:} Every arrival results in an increase of the system state dimension, and every departure in a decrease of the system state dimension. Analyzing the evolution of the system state is therefore much more challenging than in \quotes{closed systems}. 

\emph{Absence of usual \quotes{convergence:}} Being continuously perturbed by departures and arrivals, open systems will never asymptotically converge to a specific state (this is clear from Figure~\ref{fig:fixed_system_size_with_replacement_small}). Rather, they may approach some form of steady state behavior, which can be characterized by some relevant descriptive quantities. As in classical control in the presence of perturbations, the choice of the measures is not neutral, and different descriptive quantities may behave in very different way.

\begin{figure}[!htbp]
\begin{center}
\includegraphics[scale=0.44,clip = true, trim=3cm 8cm 3cm 8.4cm,keepaspectratio]{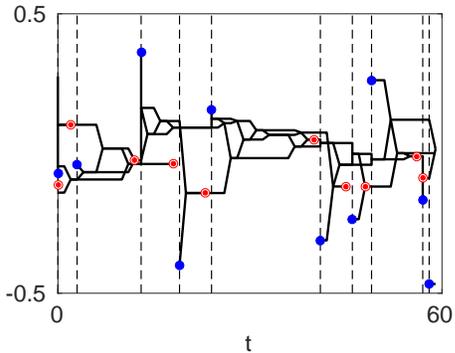}\\
\caption{\label{fig:fixed_system_size_with_replacement_small}
Example of dynamics of an open multi-agent system with random agent replacements and pairwise average gossips. The evolution with time of the agents values are represented by black continuous lines. Red circles highlight the departing agents while the blue circles correspond to the newly arrived agents. Vertical black dashed lines depict the non uniform random replacement instants. The repeated replacements prevent convergence to consensus.
See Section~\ref{sec:fixed-size-system} for a precise description of the system.
}
\end{center}
\end{figure}

\emph{Robustness and quality of the algorithms:}
Although this is not treated here, departures and arrivals also have a fundamental impact on the design of decentralized algorithms over open systems. These events will often imply a loss of information or a change in the algorithm desired result (information held by new agents may indeed affect a value that the algorithm should compute). Hence the algorithms should be robust to departures and arrivals. On the other hand, algorithms over open systems cannot be expected to be \quotes{exact}: When the system composition keeps changing, algorithms able to maintain an approximate answer most of the time may be preferable to those that would eventually provide an exact answer if the system composition were to remain constant.

\subsection{Contribution}

We provide results which are part of an ongoing study on open multi-agent systems. Here, we study systems where the interactions among agents occur in discrete-time and take the form of an average pairwise gossip algorithm \cite{boyd2006randomized}. We assume all-to-all (possible) communications, focusing on systems where departures and arrivals take place at random times, see Section~\ref{sec:sys-description} for a complete definition.

We analyze the system evolution in terms of two \quotes{scale-independent} quantities. We find that the first two expected moments form such relevant quantities : namely we study 
the expected square mean $\E(\bar x ^2)$ and the expected mean square $\E(\overline {x^2})$ of the system state $x$. These quantities also provide the evolution of the expected variance $\E(\overline {x^2} - \bar x^2)$. We show in Section~\ref{sec:effect_operations} that these quantities can be characterized exactly, and that they evolve according to an associated 2-dimensional linear system.

In Section~\ref{sec:fixed-size-system}, we analyze in detail the case of systems with replacements taking place at random time : a departure is immediately followed by an arrival. Each time an event occurs, it is either a replacement or a gossip step with a certain probability. In other words, between two consecutive replacements, $K$ gossip steps occur, where $K$ is a non-negative integer random variable. As part of the study, we characterize the variable $K$. We then focus in Section~\ref{sec:growing} on growing systems that agent keep joining without ever leaving. It will in particular be shown that random arrivals or replacements can result in a significant performance decrease in terms of variance.

Results on simplified versions of the systems considered here were presented at the Allerton Conference on Communication Control and Computing\cite{HendrickxMartin2016allerton}. The main differences with \cite{HendrickxMartin2016allerton} are (i) the arrivals and replacements are probabilistic events in this paper, while they followed a deterministic (and generally periodic) sequence in \cite{HendrickxMartin2016allerton}, (ii) a study of the system convergence rate and the interpretation of the corresponding eigenvectors, and (iii) a different choice of the moments studied, allowing for simpler proofs. 

\subsection{Other works on open multi-agent systems}

The possibility of agents joining or leaving the system has been recognized in computer science, and specific architectures have for example been proposed to deploy large-scale open multi-agent systems, see e.g. THOMAS project \cite{carrascosa2009service}. There also exist mechanisms allowing distributed computation processes to cope with the shut down of certain nodes or to take advantage of the arrival of new nodes. 

Frameworks similar to open multi-agent arrivals have also been considered in the context of trust and reputation computation, motivated by the need to determine which arriving agents may be considered reliable, see e.g. the model FIRE \cite{huynh2006integrated}. However, the study of these algorithms's behavior is mostly empirical.

Varying compositions were also studied in the context of self-stabilizing population protocols \cite{angluin2008self,delporte2006birds}, where interacting agents (typically finite-state machines) can undergo temporary or permanent failures, which can respectively represent the replacement or the departure of an agent. The objective in those works is to design algorithms that eventually stabilize on the desired answer if the system composition stops changing, i.e. once the system has become \quotes{closed}.

Opinion dynamics models with arrivals and departures have also been \modif{empirically} studied in \cite{torok2013opinions,iniguez2014modeling}.

\section{System Description}\label{sec:sys-description}

We consider a multi-agent system whose composition evolves with time. We use integers to label the agents, denote by $\NN(t)\subset \N$ the set of agents present in the system at time $t$, and by $n(t)$ the number of agents present at time $t$, i.e. the cardinality of $\NN(t)$. Each agent $i$ holds a value $x_i(t)\in \Re$, and we make no assumptions about the values held at $t=0$ by the agents initially present in the system. 

We consider a discrete evolution of the time $t\in \N$. It is possible to interpret the discrete time $t$ as a sampling of a continuous time variable. Samples then correspond to instants where an event occurred. We will comment later on this interpretation and on its implication on the scaling of different parameters. At each time $t$, one of three events may occur:

\emph{(a) Gossip:} Two agents $i,j\in \NN(t)$ are uniformly randomly and independently selected among the $n(t)$ agents present in the system (with in particular the possibility of selecting twice the same agent), and they update their values $x_i,x_j$ by performing a pairwise average:
\begin{equation}\label{eq:gossip}
x_i(t+1) = x_j(t+1) = \frac{x_i(t)+x_j(t)}{2}.
\end{equation}

\emph{(b) Departure:} One uniformly randomly selected agent $i\in \NN(t)$ leaves the system, so that $\NN(t+1) = \NN(t)\setminus \{i\}$ and $n(t+1) = n(t) - 1$. This event may only occur if $n(t)>0$.

\emph{(c) Arrival:} One \quotes{new} agent $i\not \in \NN(s)$, $\forall s\leq t$, joins the system, so that $\NN(t+1) = \NN(t) \cup \{i\}$ and $n(t+1) = n(t)+1$. The initial value $x_i(t+1)\in \Re$ of the arriving agent is drawn independently from a constant distribution $\DD$ with mean $0$ and variance $\sigma^2$. (The results can immediately be adapted to systems where the mean of the arriving agents states is an arbitrary constant).

In addition, we will sometimes consider for simplicity a \quotes{replacement} event, which consists of the instantaneous combination of a departure and an arrival: \modif{an agent leaves the system and is instantaneously replaced.}

Note that all the random events above are assumed independent of each other.

\subsection*{\modif{Scale-independent quantities of interest}}\label{sec:scale-independent-quantities-of-interest}

The aim of the study is to characterize the disagreement among agents, \ie the distance to consensus. We say that consensus is reached asymptotically when
\begin{equation}\label{eq:class_consensus}
\lim_{t \rightarrow \infty} \,  \max_{(i,j) \in \NN(t)^2} |x_i(t) - x_j(t)| = 0.
\end{equation}
If the system dynamics does not include agent departures or arrivals, it is known that the gossip process we consider leads to consensus, see e.g.\cite{FagnaniZampieri2007,boyd2006randomized}. The objective here is to understand how agent arrivals and departures impact the disagreement among agents.
To do so, we study several quantities of interest. Because the system size may change significantly with time, we focus on scale-independent quantities, i.e. quantities whose values is independent of the size of the system. We consider in particular 
the empirical mean of the squares
and the variance defined as
\begin{equation}\label{eq:def-mean-variance-xk-elementwise} 
\begin{array}{rl}
&\overline{x^2} = \frac{1}{n} \ssum_{i \in \NN} x_i^2, \\
 &\Var(x)= \frac{1}{n} \ssum_{i \in \NN} (x_i - \bar{x})^2 = \overline{x^2} - (\bar{x})^2,
\end{array}
\end{equation}
respectively, where references to time were removed to lighten the notation. Our study will focus on the evolution of $\E \Var(x)$, which will also require monitoring $\E (\bar x)^2$ and $\E \overline{x^2}$. When new agents keep arriving it is impossible to achieve asymptotic consensus in the sense of \eqref{eq:class_consensus}, because the new agent's value will with high probability be different from the value of the agents already present in the system. The study of $\E \Var(x)$ will allow us to see how \quotes{far} the system will be from consensus. But we will see that in \modif{certain systems whose sizes grow unbounded}, we may have $\lim_{t\to\infty} \E \Var(x) =  0$, corresponding to a form of \quotes{almost consensus}. The expected mean $\E \bar{x}$ could also have been monitored. It evolves following an independent one-dimensional linear system. However, we skip this part of the study due to space limitations.

\section{Effect of different operations}\label{sec:effect_operations}

In the sequel we will show that the evolution of the expected moments $\E(\bar{x}^2)$ and $\E\overline{x^2}$ is governed by an affine system from which we will derive the evolution of $\E \Var(x)$. For notational simplicity we denote by $X$ the vector containing $\bar x^2$ and $\overline {x^2}$ so that $\Var(x) = (-1,1) X$.

\begin{lemma}[Gossip step]\label{lem:gossip_moment}
Suppose that a randomly selected pair of agents engage in a gossip averaging according to equation~\eqref{eq:gossip}. Let $x$ be the state of the system before that interaction, $x'$ its state after the interaction and $n$ the number of agents. There holds
\begin{equation}\label{eq:matrix_gossip_moment}
\E X' = A_g \E X,
\text{ where }
 A_g = \prt{\begin{array}{cc} 1 &  0 \\ \frac{1}{n} & 1 - \frac{1}{n} \end{array}},
\end{equation}
and as a consequence
\begin{equation}\label{eq:variance_gossip}
\E \Var(x') = \left(1- \frac{1}{n}\right) \E \Var(x). 
\end{equation}
\end{lemma}
Notice that this result or its variations are available in many previous works, but its proof is presented for the sake of completeness.

\begin{proof}
let us first fix the nodes $i,j$ involved in the gossip. Observe  that $x_i'+x_j' = 2\frac{x_i+x_j}{2}=x_i+x_j$, and that $x_k'=x_k$ for all $k\neq i,j$. Hence $\bar x' = \bar x$, which establishes the fist line of \eqref{eq:matrix_gossip_moment}. For the second line, since $x_k=x'_k$ for every $k\neq i,j$, there holds
\begin{align}\label{eq:gossip_avg_x2}
\overline{{x'}^2} &= \frac{1}{n}\sum_{k=1}^n{x'}_k^2 = \overline{x^2} + \frac{1}{n}\prt{ 2\prt{\frac{x_i+x_j}{2}}^2- x_i^2-x_j^2} \nonumber \\
&= \overline{x^2} + \frac{1}{n}\prt{x_ix_j - \frac{1}{2}x_i^2-\frac{1}{2}x_j^2}
\end{align}
Observe that $E(x_i^2|x)= E(x_j^2|x)= \overline{x^2}$ and $E(x_ix_j|x)=\bar x^2$. Taking the expectation with respect to $i$ and $j$ in \eqref{eq:gossip_avg_x2} yields
\begin{align*}
\E( \overline{{x'}^2}| x)
= \prt{1-\frac{1}{n}}\overline{x^2} + \frac{1}{n} \bar x^2, 
\end{align*}
from which the second line of \eqref{eq:matrix_gossip_moment} follows.
\end{proof}

\begin{lemma}[Arrival of the $n+1$-th agent]\label{lem:arrival_moments_n+1}
Suppose that an agent arrives into the system bringing the number of agents from $n$ to $n+1$. Denote $x$ the state before arrival, $x'$ the state after arrival. Then, there holds
\begin{align}\label{eq:system-arrival}
\E X'  &= A_a \E X + b_a
\end{align}
where
\[
A_a = \prt{\begin{array}{cc} 
\frac{n^2}{(n+1)^2}& 0\\
0& \frac{n}{n+1}
\end{array}} \text{ and } 
b_a = \sigma^2 \prt{\begin{array}{c} \frac{1}{(n+1)^2} \\ \frac{1}{n+1}
\end{array}}
\]
\end{lemma}

\begin{proof}
We label $n+1$ the arriving agent for simplicity, so that $x'_k=x_k$ for all $k\leq n$.  We begin by computing the new average :
\begin{align}\label{eq:arrival_average}
\bar x' &= \frac{1}{n+1}\prt{x'_{n+1} + \sum_{k=1}^{n}x_k  }= \frac{n}{n+1} \bar x + \frac{1}{n+1}x'_{n+1}.
\end{align}
Since $\E x'_{n+1} = 0$, we have $\E(\bar x'|x) = \frac{n}{n+1}\bar x$. By exactly the same reasoning but using $\E {x'_{n+1}}^2 = \sigma^2$  we also obtain
\begin{align}\label{eq:arrival_x^2}
\E(\overline{x'^2}|x) &= \frac{n}{n+1}\overline{x^2} + \frac{1}{n+1}\sigma^2,
\end{align}
from which the second line of \eqref{eq:system-arrival} follows.
Turning to the first line, 
we obtain from \eqref{eq:arrival_average}
\begin{align*}
\E((\bar x')^2|x) =& \frac{1}{(n+1)^2}\prt{n^2(\bar x)^2 + n \bar x \E x'_{n+1}+\E (x'_{n+1})^2} \\
=& \frac{n^2}{(n+1)^2} (\bar x)^2 + 0  +\frac{1}{(n+1)^2} \sigma^2.
\end{align*}
\end{proof}

\begin{lemma}[Departure]\label{lem:departure_moment}
Suppose that a randomly selected agent departs from the system. Denote $x$ the state before departure, $x'$ the state after departure and $n\geq 2$ the number of agents before departure. Then, there holds
\begin{equation}\label{eq:departure}
\E X'=  \prt{\begin{array}{cc} \frac{n^2-2n}{(n-1)^2} & \frac{1}{(n-1)^2}\\ 0 & 1\end{array}} \E X.
\end{equation}
\end{lemma}

\begin{proof}
Let $j$ be the randomly selected agent that leaves the system. The mean is modified as
\begin{align}\label{eq:mean_departure}
\bar x' = \frac{1}{n-1}\prt{\prt{\sum_{k=1}^n x_k} - x_j}
= \frac{1}{n-1}\prt{n\bar x - x_j}.
\end{align}
By exactly the same reasoning, there holds
$\overline {x'^2}= \frac{1}{n-1}\prt{n \overline{x^2} - x_j^2}.$
Since $j$ is randomly selected, $\E(x_j^2|x) = \overline{x^2}$. Hence, 
\begin{align*}
\E (\overline {x'^2}|x) = \frac{1}{n-1}\prt{n \E \overline{x^2} - \E \overline{x^2}}=\E \overline {x^2},
\end{align*}
which implies the second line of \eqref{eq:departure}. For the first line, taking into account $\E(x_j|x) = \bar x$, it follows from \eqref{eq:mean_departure} that
\begin{align*}
\E ( (\bar x')^2|x)
&= \frac{1}{(n-1)^2}\prt{n^2(\bar x)^2 - 2n(\bar x)\E(x_j|x) + \E(x_j^2|x)}\\
&=  \frac{n^2-2n}{(n-1)^2}(\bar x)^2  + \frac{1}{(n-1)^2} \overline{x^2}.
\end{align*}
\end{proof}

We now consider the replacement of an agent, which consists of a departure immediately followed by an arrival. The next result follows from a combination of Lemma \ref{lem:departure_moment} and \ref{lem:arrival_moments_n+1}, the latter applied to a system of size $n-1$ joined by a $n^{th}$ agent.

\begin{lemma}[Replacement]\label{lem:replacement_matrix_moment}
Suppose that a randomly selected agent departs from the system and is immediately replaced by a new agent, leaving the size of the system unchanged. Denote $x$ the state before replacement, $x'$ the state after replacement and $n$ the constant number of agents. Then, there holds
\begin{equation}\label{eq:exp_replacement}
\E X' = A_r \E X + b_r,
\end{equation}
where
\[
A_r = \prt{\begin{array}{cc} 
\frac{n-2}{n}& \frac{1}{n^2}\\
0& \frac{n-1}{n}
\end{array}} \text{ and } b_r = \prt{\begin{array}{c} \frac{\sigma^2}{n^2} \\ \frac{\sigma^2}{n}
\end{array}}.
\]
\end{lemma}

\section{Fixed-size system with random replacement}\label{sec:fixed-size-system}

\subsection{Model and fixed points}\label{sec:probabilistic-model}

We assume now that at each time step, there is a probability $p$ that an agent is replaced, and a probability $1-p$ that a gossip step takes place. After giving an exact value for the fixed point of the expected system behavior, we will provide asymptotic results when the system size $n$ is large. Note that when considering results for large $n$, it is natural to keep $p$ constant, although other forms of scaling could be considered. Suppose indeed that our discrete times correspond to the sampling of a continuous-time process at those times at which an event occurs. We suppose that the rates of gossip and of departure (leading to a replacement) of a single agent are both independent of the system size, which would be natural for large systems. As a result, the rates of gossip and replacements at the system level both scale linearly with $n$, so that the probability $p$ that a randomly selected event is a replacement remains independent of $n$.

The following result, based on Lemma \ref{lem:replacement_matrix_moment} and equation \eqref{eq:matrix_gossip_moment}, describes the expected evolution of the system. 
\begin{theorem}\label{th:linear-system-random-event}
Suppose that an event occurs. This event is an agent replacement with probability $p$ or gossip with $1-p$.
Denote $x$ the state before the event, $x'$ the state after. Then, there holds
\begin{equation}\label{eq:proba_iteration1}
\E X' = \prt{\begin{array}{cc}
1 -\frac{2p}{n}& \frac{p}{n^2}\\
\frac{1-p}{n}& 1- \frac{1}{n}
\end{array}} \E X + \sigma^2  \prt{\begin{array}{c} \frac{p}{n^2} \\ \frac{p}{n}
\end{array}}.
\end{equation}
\end{theorem}

\begin{proof}
Since the probability of the events are independent of $x$, the conditional expected value is computed as follows :
\begin{align*}
\E(X'|x)  & = (1-p) \E(X'|\text{gossip},x) + p \E(X'|\text{repl.},x)\\&= 
(1-p) A_g X + p (A_r X + b_r)\\
&= ((1-p) A_g + p A_r) X + p b_r.
\end{align*}
Therefore, there holds
$
\E X' = ((1-p) A_g + p A_r) \E X + p b_r,
$ which yields \eqref{eq:proba_iteration1}.
\end{proof}

One can verify that the fixed point of \eqref{eq:proba_iteration1} is
\begin{align}\label{eq:fixed_moment}
\E (\bar x)^2\rvert_{eq} &= \frac{p+1}{p + 2n-1}\sigma^2\nonumber\\
\E \overline {x^2}\rvert_{eq} &= \frac{1+p(2n-1)}{p+2n-1}\sigma^2
\end{align}
leading to a variance $\E \Var(x)\rvert_{eq} = \sigma^2 \frac{2p(n-1)}{p+2n-1}$.

The asymptotic values of these expressions admit some interpretation. Suppose first that $p=1$, meaning that no gossip ever takes place. We obtain then a variance $\E \Var(x)\rvert_{eq} =\sigma^2 (1-\frac{1}{n})$, and an expected square of the average $\E({\overline x})^2\rvert_{eq} = \frac{\sigma^2}{n}$ consistently with a process where agents are just replaced, \textit{i.e.}, a system eventually consisting of agents with $n$ random i.i.d. values with mean 0 and variance $\sigma^2$. (This is also the fixed point of the affine equation in Lemma~\ref{lem:replacement_matrix_moment}). For $p\to 0$, the number of gossips steps between two replacements tends to infinity, so that a perfect averaging takes place before any replacement. We obtain in that case a variance $\E\Var(x)\rvert_{eq} =0$, and an expected square average $\E({\overline x})^2\rvert_{eq} = \sigma^2\frac{1}{2n-1}$. This latter number is lower than what would be obtained by averaging $n$ i.i.d. values. This is because it actually results from a weighted average of the values of all agents having been part of the system the system at some present or past time. See Section IV.c of \cite{HendrickxMartin2016allerton} for a detailed computation of this value in a system with deterministic replacements.

For large $n$ and constant $p$, which we have argued above is a natural scaling, the expected square $\E({\overline x})^2\rvert_{eq}$  goes to 0, while the variance $\E\Var(x)\rvert_{eq}$ goes to $\sigma^2 p$. This result is parallel with that obtained in \cite{HendrickxMartin2016allerton} for periodic replacement, taking into account that the average number of gossip steps between two replacements is $\frac{1}{p}-1$. We will come back to this value $\sigma^2p$ in Section \ref{sec:growing}.

To illustrate Theorem~\ref{th:linear-system-random-event}, we consider an open system with random events (replacements or gossip steps) 
with
$n=25$ agents and replacements occur with probability $p=0.05$ while gossip steps occur with probability $1-p$. The system has evolved until it has reached $100$ replacements. Arriving agent values are drawn uniformly in $[-\frac{1}{2},\frac{1}{2}]$ so that $\sigma^2=\frac{1}{12}$. Figure~\ref{fig:fixed_system_size_with_replacement} displays a realization of the trajectories along with the expected dynamics for the scale-invariant quantities. In the top plot, it appears that agents leaving the system (in red) tend to have a more moderate state compared to agents arriving in the system (in blue). This is due to the gossip steps. Despite not leading to a consensus, the open system still presents a contracting tendency.
Besides, as seen in the bottom plot, for a sufficient number of agents (here 25), the expected variance rather well approximates the dynamics of empirical variance realization. Also for this number of agents, the square mean remains small (of order $10^{-3}$) compared to the mean square, as a consequence, the variance is mainly due to $\overline{x^2}$. This would not be the case for $n=5$ agents for instance.

The illustration provided in the introduction (Figure~\ref{fig:fixed_system_size_with_replacement_small}) was obtained for an open multi-agent system of the same kind with $n=4$ agents and a replacement probability $p=0.1$ where the dynamics consider the system up to $10$ replacements.

\begin{figure}[!htbp]
\begin{center}
A realization of trajectories
\includegraphics[scale=0.44,clip = true, trim=3cm 8cm 3cm 8.4cm,keepaspectratio]{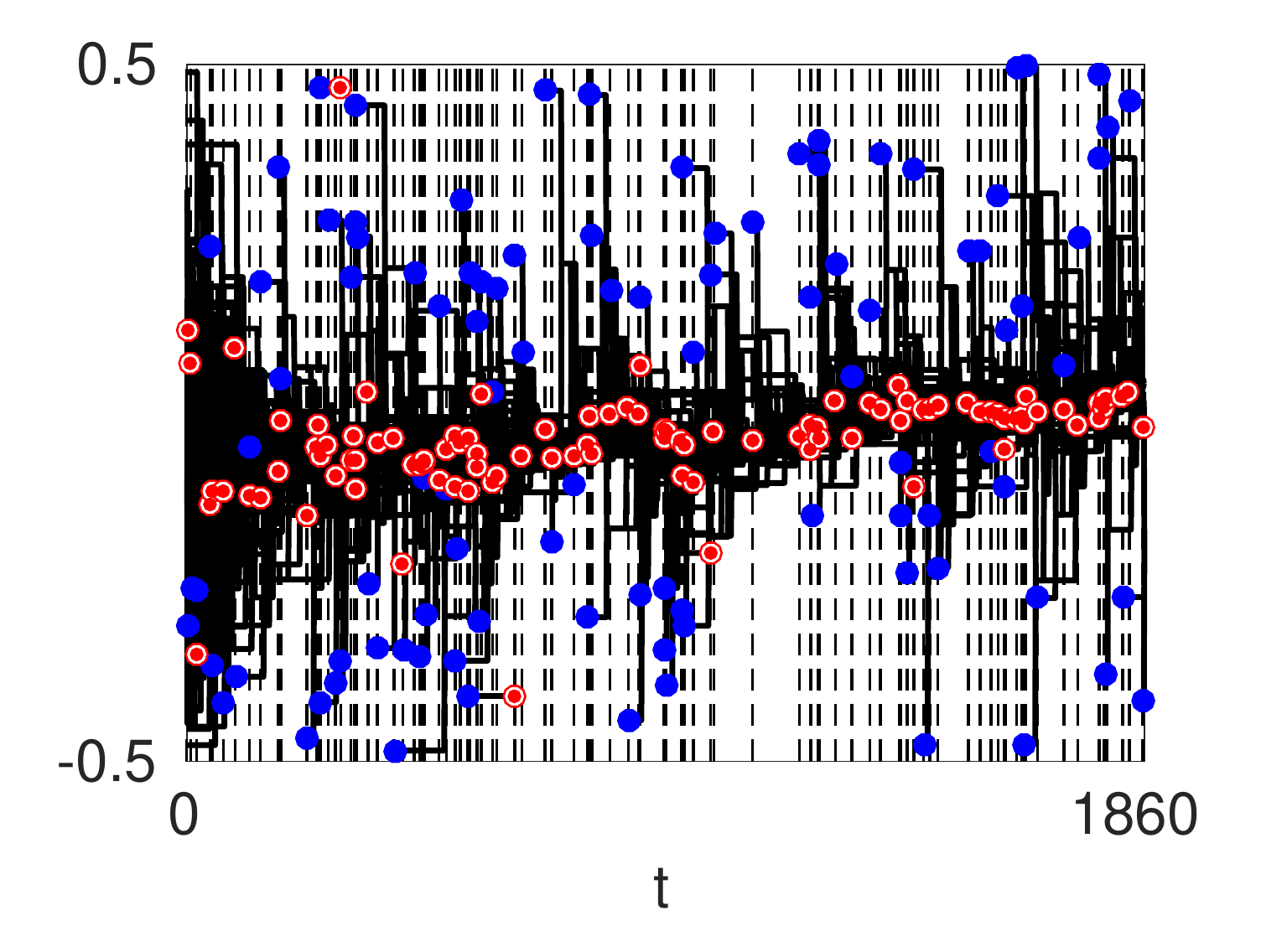}\\
\includegraphics[scale=0.44,clip = true, trim=3cm 8cm 3cm 8.4cm,keepaspectratio]{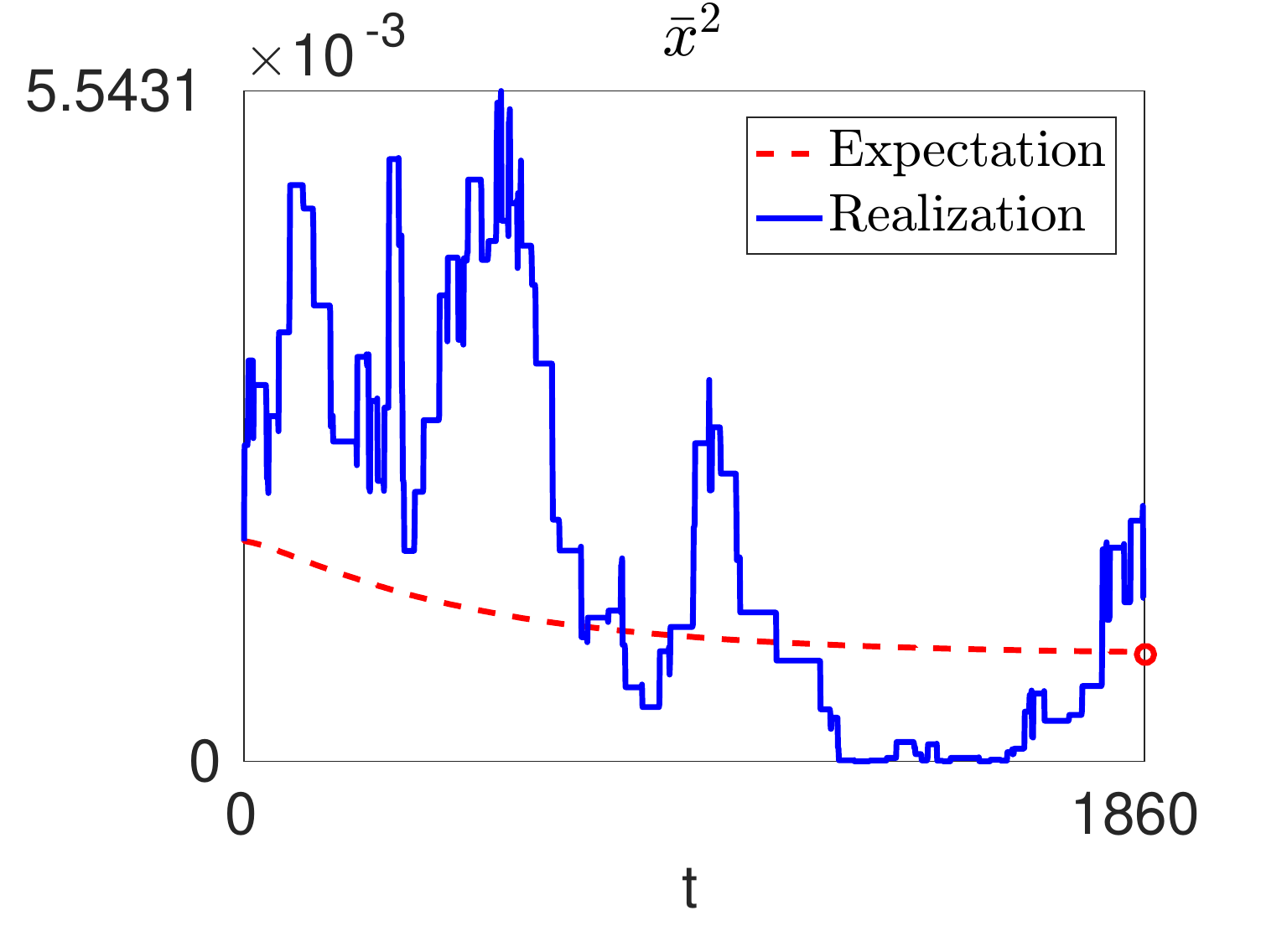}\\
\includegraphics[scale=0.44,clip = true, trim=3cm 8cm 3cm 8.4cm,keepaspectratio]{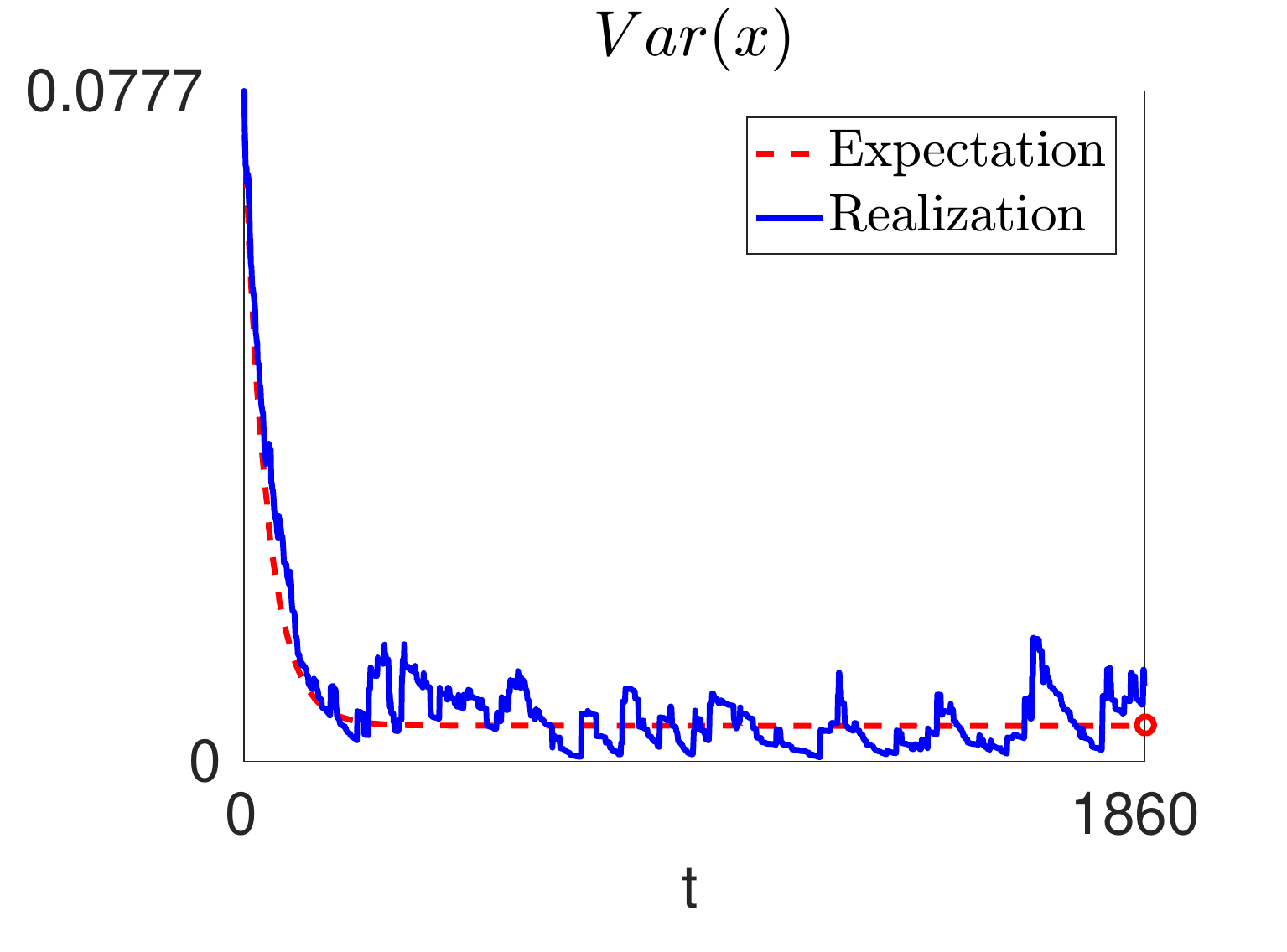}\\
\caption{\label{fig:fixed_system_size_with_replacement}
Illustration of an open system with random events.
(top) shows the evolution with time of the agents values (in black) for a typical realization. Red circles highlight the departing agents while the blue circles correspond to the newly arrived agents. Vertical black dashed lines depict the non uniform random replacement instants. (middle) shows the evolution of the square mean value (realization in continuous blue line, expectation in dashed red line). (bottom) shows the evolution of the expected variance (realization in continuous blue line, expectation in dashed red line). Expectation were computed using Theorem \ref{th:linear-system-random-event}. The asymptotic values given in equations \eqref{eq:fixed_moment} are provided for the expected square mean and the expected variance in red circle at final time in the middle and bottom plots.
}
\end{center}
\end{figure}
\subsection{Convergence rate}

We now study the rate at which the expected moments will converge to the fixed points described above. Eigenvalues of the matrix in \eqref{eq:proba_iteration1} were computed on Mathematica :
\begin{align*}\label{eq:eigenvalues_fixed_size}
r_{+,-} = \frac{\left(2n - 2p  - 1 \rpm \sqrt{\Delta} \right)}{2n},
\end{align*}
with 
 $\Delta = \left(1-2p\right) ^2 + \frac{4p(1-p)}{n}.$
(The choice of notation $r_+,r_-$ 
comes from the use of either $+\sqrt{\Delta}$ or $-\sqrt{\Delta}$ in the root expression.)
For large $n$ and, since $\sqrt{\Delta} = |1-2p| + o(1)$, the eigenvalues are of order 
\[r_+ =1-\frac{2p}{n}+\ooneovern , r_- =1-\frac{1}{n}+\ooneovern,\] 
which happen to be the diagonal elements of the matrix, as if neglecting the $1/n^2$ term on the upper right hand side.

The corresponding eigenvectors are 
\[
v_+ = \left(\frac{2p-1 - \sqrt{\Delta}}{2(p-1)} , 1\right)^T = \left(\frac{2p-1}{p-1} + o(1),1\right)^T,
\]
and
\[
v_- = \left(\frac{2p-1 +  \sqrt{\Delta}}{2(p-1)} , 1\right)^T = (o(1),1)^T,
\]
unless $p=1/2$ in which case higher order term needs to be taken into account.
\\

\emph{\bf Interpretation:}

The couple $(r_-,v_-)$ is independent of $p$ for large $n$. The eigenvector $v_-$ has a vanishingly small component in $\E({\overline x})^2$ and  concerns asymptotically exclusively $\E \overline {x^2}$ which is then essentially equivalent to $\E \Var(x)$. It follows from Lemma \ref{lem:gossip_moment} that every gossip iteration will contract this quantity by $1-\frac{1}{n}$, and from Lemma \ref{lem:replacement_matrix_moment} that every replacement will also contract that quantity by $1-\frac{1}{n}$ (when neglecting $\E({\overline x})^2$) before adding a constant term. If the constant term was set to $0$, $\E \overline {x^2}$  would thus contract at a rate $1-\frac{1}{n}\simeq r_-$ at every iteration, independently of whether a gossip or replacement takes place. For large $n$, the couple $(r_-,v_-)$ can thus be related to the contraction of $\E \overline {x^2}$

The couple $r_+,v_+$, on the other hand, does strongly depend on $p$. In particular, $r_+ = 1$ if $p=0$. Remember that the square average remains unchanged when gossip iterations occur. But when a replacement occurs, it follows from Lemma \ref{lem:replacement_matrix_moment} that the new square average $\E({\overline x'})^2$ contains a contribution $\frac{n-2}{n} = (1-\frac{2}{n})$ of the previous value, some constant contribution related to the arriving agent, and a vanishingly small contribution related to the average square value (assuming that $\E({\overline x})^2$ and $\E \overline {x^2}$ are of the same order of magnitude, which is the case in $v_+$ when $p=0$). Hence, if we were to set the independent term $b_r$ at 0 in \eqref{eq:exp_replacement} corresponding to $\sigma^2=0$, the square average would be multiplied by $(1-\frac{2}{n})$ at each replacement and by 1 otherwise. Since replacement occur with a probability $p$, this yields a rate $p(1-\frac{2}{n}) + (1-p) =   1-\frac{2p}{n}\simeq r_+$.

Observe now that there is a transition at $p=\frac{1}{2}$. Indeed, for $p>\frac{1}{2}$, corresponding to frequent replacements, the largest eigenvalue is $r_-$. The rate of convergence to the steady state depends thus on the rate of variance reduction, which is independent of $p$ (for large $n$). On the other hand, for $p<\frac{1}{2}$, the convergence is generally dominated by phenomena related to the convergence of the the square average, which does depend on $p$ as this quantity only changes when a replacement takes place.

\section{Growing system without departure}\label{sec:growing}

We focus now on systems whose sizes grow unbounded because new agents keep joining while no agent ever leaves: Similarly to the fixed-size model described in Section~\ref{sec:probabilistic-model}, at each time, a new agent joins the system with probability $p_n$ or a gossip step occurs with probability $1-p_n$, where $n$ is the number of agent in the system. The discussion of the dependence of $p_n$ on $n$ is deferred to the end of the section. When no ambiguity is possible, we will simply use $p$. We assume again that the initial value of every agent when joining the system is a random variable with zero mean and variance $\sigma^2$. We denote by $t_n$ the time just after the arrival of the $n$-th agent, and we let $K_n:= t_{n+1} - t_{n}-1$ be the number of gossip steps taking place between the arrival of agents $n$ and $n+1$. Both the $t_n$ and $K_n$ are random variables, and the $K_n$ follow a
geometric distribution with parameter $p_n$ so that $\PP (K_n = k) = (1-p_n)^{k}p_n$.
Since we are only interested in expected quantities and that the set of sequences $t_n$ which are bounded has probability $0$, we assume in the sequel that sequence $t_n$ is unbounded.

We will focus on the values of the expected average of the square $\E(\overline{x^2})$ and the expected square of the average $\E(\bar{x}^ 2)$ at the times $t_n$, just after the arrivals of the $n$-th agents. 
We will see that the evolution of these values can
be described by a two-dimensional linear system. This system is here time-varying because $n$ is not constant, but the absence of departures makes it triangular, and hence easier to analyze. We first provide the evolution of the expectation of vector $X = (\bar{x}^2 , \overline{x^2})^T$ between two agent arrivals.

\begin{lemma}\label{lem:K-gossip-steps}
Let $x$ be the state of the system after the arrival of the $n^{th}$ agent, and suppose that at every time-step, an arrival takes place with probability $p$ and a gossip with probability $1-p$. Let then $x'$ be the state after the arrival of agent $n+1$, so after the first arrival event. There holds
$\E X' = A_G \E X$
where
\[
A_G = \prt{\begin{array}{cc}
1& 0\\
1 - \gamma & \gamma
\end{array}} \text{ with } \gamma = \frac{n}{n - 1 + \frac{1}{p}}.
\]
\end{lemma}

\begin{proof}
To obtain the expectation of vector $X'$ after the gossip, we condition with regards to $K_n$ which is by definition the number of gossip steps having taken place. By Lemma~\ref{lem:gossip_moment}, there holds
\begin{align*}
\E(X'|x)  & = \ssum_{k=0}^{+\infty} \E(X'|K_n = k,x)\PP(K_n = k) \\
&= \ssum_{k=0}^{+\infty} A_g^k X (1-p)^kp
= p\ssum_{k=0}^{+\infty} (1-p)^k A_g^k X \\
&= p(I - (1-p)A_g)^{-1} X = A_G X.
\end{align*}
Recalling that $\E(\E(X'|x)) = \E(X')$, the expectation of the previous equation provides the result.
\end{proof}
A consequence of this lemma is that between two agent arrivals, the expected variance evolves autonomously as
\begin{equation}\label{eq:Var-after-gossip-steps}
\E \Var(x') = \gamma \E \Var(x). 
\end{equation}
Applying Lemma~\ref{lem:arrival_moments_n+1}, we obtain the evolution of the variance after agent arrivals.
\begin{lemma}
Recall that $x(t_n)$ is the state of the system just after the arrival of the $n$-th agent and $x(t_{n+1})$ the state just after the arrival of the $n+1$-th agent, so that $x(t_{n+1})$ is obtained starting from state $x(t_n)$ and applying $K_n$ gossip steps and an arrival. Then, there holds
\begin{align}
\E(\bar{x}^2(t_n)) &= \frac{\sigma^2}{n}, \label{eq:proba_sq}\\
(n+1)\E\Var(x(t_{n+1})) &= \gamma n\E\Var(x(t_{n})) + \sigma^2.\label{eq:proba_iteration3}
\end{align}
\end{lemma}

\begin{proof}
 Applying first Lemma~\ref{lem:K-gossip-steps} and then Lemma~\ref{lem:arrival_moments_n+1}, we obtain
 \begin{equation*}\label{eq:linear-system-AaAg}
 \E X(t_{n+1}) = A_a A_G \E X(t_{n}) + b_a.
 \end{equation*}
 Since $A_a$ is diagonal and $A_G$ is lower triangular, the first line of the previous equation provides
 \[
 \E(\bar{x}^2(t_{n+1})) = \frac{n^2}{(n+1)^2} \E(\bar{x}^2(t_{n})) + \frac{\sigma^2}{(1+n)^2},
 \]
 which, recalling that initially $\E(\bar{x}^2(0)) = 0$, grants equation~\eqref{eq:proba_sq} by induction on $n$.
 To obtain equation~\eqref{eq:proba_iteration3}, recall that $\Var(x) = (-1, 1)X$ so that
$\E\Var(x(t_{n+1}))$
 \begin{align*}
 =& (-1, 1) (A_a A_G \E X(t_{n}) + b_a) \\
 =&  \left(-\frac{n^2}{(n+1)^2} + (1-\gamma)\frac{n}{n+1} , \gamma \frac{n}{n+1}\right) \E X(t_{n}) \\&+ (-1, 1) b_a \\
 =& \left(-\frac{n^2}{(n+1)^2} + \frac{n}{n+1}\right)\E(\bar{x}^2(t_{n})) + \frac{n\gamma}{n+1} \E\Var(x(t_{n})) \\
 &+ (-1, 1) b_a \\
 =& \frac{n}{(n+1)^2}\E(\bar{x}^2(t_{n})) + \frac{n\gamma}{n+1}\E\Var(x(t_{n})) + \frac{n}{(n+1)^2} \sigma^2.
\end{align*}
We conclude using equation~\eqref{eq:proba_sq}.
\end{proof}

This recursion allows obtaining the following theorem characterizing the asymptotic variance, and proved in Appendix \ref{appen:proof_growing}.
\smallskip
\begin{theorem}\label{thm: growing system}
Consider the growing system without departure, and remember that when the system with $n$ agents undergoes an event, $p_n$ describes the probability of the $n+1$-th agent arrival rather than a gossip step.

\noindent (i) If $p_n = p > 0$ for all $n\geq n_0$ for some $n_0$, then 
$$\lim_{n\to \infty} \E \Var(x(t_n)) = p\sigma^2.$$

\noindent (ii) If $\lim_{n\to \infty} p_n = 0$, then $\lim_{n\to \infty} \E \Var(x(t_n)) = 0$.
\end{theorem}
\bigskip

Theorem \ref{thm: growing system}(ii) shows that the system essentially converges to a consensus as soon as $p_n$ goes to zero, even if this convergence is very slow. This can also be interpreted in terms of $\E(K_n)$, the expected number of gossip steps between two arrivals. Since $K_n$ follows a geometric law of parameter $p_n$, $\E(K_n) = \frac{1-p_n}{p_n}$ which diverges to $+\infty$ when $p_n$ goes to zero. As a consequence, the system converges to a consensus as soon as the expected number of gossip steps between two consecutive arrivals diverges, and even if the number $\E(K_n)/n$ of gossips per agent between two consecutive arrivals tends to 0. Note, however, that each agent gets
involved (with probability 1) in infinitely many gossips when $\E(K_n)\to\infty$. The expected number of gossips in which an agent has been involved at time $t_n$ is indeed $2\frac{1}{n}\sum_{m=1}^n \E(K_m)$,
which grows unbounded.

By contrast, in the case of a fixed probability $p_n = p$ corresponding to a fixed $\E(K_n)=:K$ with $p=\frac{1}{K+1}$, agents have on average been involved in 
$
2\frac{1}{n}\sum_{m=1}^n \E(K_m)= 2K
$
gossips after any given arrival, which intuitively explains why the variance stays bounded away from 0. But the actual asymptotic value $ p\sigma^2 = \frac{\sigma^2}{K+1}$ obtained in Theorem \ref{thm: growing system}(i) is remarkably high, and is actually the same as that obtained for the fixed-size system in Section \ref{sec:probabilistic-model} for large $n$.  As a basis for comparison, suppose we had first waited until the $n$ independent agents were present in the system, which would yield an expected variance $\sigma^2\frac{n-1}{n}$, and then performed the same number $nK$ of gossip averaging operations between randomly selected pairs of nodes. It follows from application of equation~\eqref{eq:variance_gossip} that the expected variance would then have been
$$
\frac{n-1}{n}\sigma^2 \left(1-\frac{1}{n}\right)^{nK}\to_{n\to \infty} \sigma^2 e^{-K},
$$
which is significantly lower than $\sigma^2/(K+1)$ (for $K= 5$, the ratio of variance would be $\frac{e^{-5}}{1/6}\simeq 0.04$). The dynamics of the system composition deteriorates thus considerably the performances in terms of variance reduction.

We now propose possible interpretations of the evolutions of $p_n$ and $\E(K_n)$ with the number of agents $n$.
Suppose that we interpret our discrete $t$ as the sampling of a real continuous time variable $\tau$ at those times $\tau_t$ at which an event occurs. It is again reasonable to assume the interaction rate of an agent to be independent of the system size, so that the total number of gossips per unit of time $\tau$ would grow linearly with $n$, as say $\lambda_g n$. Suppose first that the agents arrive at a fixed rate $\lambda_a$. In that case, the probability of agent arrival is $\lambda_a/(\lambda_a+\lambda_g n)$ and the number of gossips between two arrivals would be linearly growing with $n$ and $\E(K_n) = n \lambda_g/\lambda_a$. Theorem \ref{thm: growing system}(ii) shows then that the variance would converge to 0.

But one could also imagine a linearly growing rate of arrivals $\lambda_r n$. This would for example be the case if the system attraction were growing with its size or if the arrivals resulted from some form of reproduction process. 
The probability of agent arrival would then be constant $p_n = p = \lambda_a/(\lambda_a+\lambda_g)$ and so will be the number of gossip iterations between two arrivals $\E K_n = K = (\lambda_g n)/(\lambda_r n)$, leading to a finite variance $p\sigma^2$.

\section{Conclusions}

We have made first steps in the analysis of open multi-agent systems, where agents can leave and arrive. We have focused on analysing open systems subject to 
a classical multi-agent algorithm :
on all-to-all pairwise gossips. We have shown that these systems could be characterized by fixed-size linear systems in terms of some of the moments. Interestingly, we have also observed that the open character of the system may result in a significant performance reduction in terms of variance reduction. 

Ongoing works include the generalization of this approach to systems where arrivals and departure follow more complex patterns, or to more complex interactions, such as gossips restricted to a graph. 

Another challenge left untackled so far is the characterization of the variability for individual realizations. Our results characterize for instance the expected value of the disagreement in the system $ \E\Var(x(t))$ where $\Var(x(t)) = \frac{1}{n(t)} \sum_{i\in\NN(t)}(x_i(t)-\bar x(t))^2$, but do not directly allow deducing the width of the probability distribution of $\Var(x(t))$.

\begin{appendix}
 
\subsection{Proof of Theorem \ref{thm: growing system} }\label{appen:proof_growing}

We prove the following proposition, which implies Theorem \ref{thm: growing system} when applied to $W_n = n \E \Var(X(t_{n}))$.
\smallskip
\begin{proposition}
Let $n_0 \ge 2$. Let $\gamma_n = \frac{n}{n-1+\frac{1}{p_n}}$ with $p_n \in (0,1)$. Let $p \in (0,1)$ and consider the sequence defined by $W_1= 0$ and 
\begin{equation}\label{eq:recur_W}
W_{n+1}  = W_n \gamma_n + \sigma^2.
\end{equation}
a) If $p_n \leq p$ for all $n\geq n_0$, then $\limsup_{n\to\infty } \frac{W_n}{n}\leq p\sigma^2$\\
b) If $p_n \geq p$ for all $n\geq n_0$, then $\liminf_{n\to\infty } \frac{W_n}{n}\geq p\sigma^2$.
\end{proposition}
\smallskip
\begin{proof}
We prove the statement (a) of the proposition. Statement (b) can be obtained in a similar way. It follows from \eqref{eq:recur_W} that 
\begin{align}\label{eq:sep_Wn}
W_n = W_{n_0}\Pi_{m=n_0}^{n-1}\gamma_m  + \sigma^2 \sum_{s=n_0+1}^{n} \Pi_{m=s}^{n-1}\gamma_m,
\end{align}
with the convention $\Pi_{m=n}^{n-1} \gamma_m = 1$.

We use notation $q = \frac{1}{p}-1$ and $q_n = \frac{1}{p_n}-1$ so that $\gamma_n = \frac{n}{n+q_n} = 1 - \frac{q_n}{n+q_n}$. Since $p_n \le p$,
$q_n \ge q > 0$ and there holds
\[\gamma_n \le 1-\frac{n}{n+q}.\]
Denote $x_n = \frac{n}{n+q}$.
Using $\log (1-x_n) \leq -x_n$ (for $x_n<1$), we obtain for $s\geq n_0$
\begin{align}\label{eq:bound_log}
\log\prt{\Pi_{m=s}^{n-1}\gamma_m} & \leq -\sum_{m=s}^{n-1} x_m = -q\sum_{m=s+q}^{n+q-1} \frac{1}{m}.
\end{align}
Observe that $\sum_{m=s+q}^{n+q-1}\frac{1}{m}$ is an upper approximation of the integral $\int_{x=s+q}^{n+q}\frac{1}{x}dx$, and hence 
$$
\frac{1}{s+q} + \dots  + \frac{1}{n+q-1} \geq \log (n+q) -\log (s+q) \geq \log \frac{n+q}{s+q}.
$$

Reintroducing this in \eqref{eq:bound_log} yields 
\begin{align*}
\Pi_{m=s}^{n-1}\gamma_m  \leq \prt{\frac{s+q}{n+q}}^q, \hspace{.5cm}\forall s\geq n_0,
\end{align*}
and hence, noticing $W_n \ge 0$, \eqref{eq:sep_Wn} implies
\begin{align}\label{eq:Wn-intermediary-approx}
W_n &\leq W_{n_0} \prt{\frac{n_0+q}{n+q}}^q + \frac{\sigma^2}{(n+q)^q} \sum_{s=n_0+1}^{n} (s+q)^q.
\end{align}
By a change of variable, the sum in the last term can be is rewritten as
\begin{align*}
 \sum_{s=n_0+1}^{n} (s+q)^q &= \sum_{s=n_0+q+1}^{n+q} s^q
 \le \sum_{s=0}^{n+q} s^q \\&\le
 \int_{0}^{n+q+1} x^q dx
 \le \frac{(n+q+1)^{q+1}}{q+1}
\end{align*}
Introducing this in equation~\eqref{eq:Wn-intermediary-approx}, we obtain
\begin{align}\label{eq:bound_W_n_almost_final}
W_n \leq W_{n_0} \prt{\frac{n_0+q}{n+q}}^q + \frac{\sigma^2}{(n+q)^q}\frac{(n+q+1)^{q+1}}{q+1}.
\end{align}
The first term in \eqref{eq:bound_W_n_almost_final} decays to 0 when $n$ grows. Hence $\limsup \frac{W_n}{n} \leq \frac{\sigma^2}{q+1} = p\sigma^2$.

Part (b) follows a parallel reasoning using the upper bound $\log(1-\frac{1}{n}) \geq -\frac{1}{n-1}$.
\end{proof}

\end{appendix}

\end{document}